\documentclass[a4paper, 11pt]{article}
\usepackage{srcltx,graphicx}
\usepackage{amsmath, amssymb, amsthm}
\usepackage{color}
\usepackage{lscape}
\usepackage{multirow}
\usepackage[mathscr]{eucal}
\usepackage[dvips]{hyperref}
\usepackage[hang]{subfigure}

\newtheorem{theorem}{Theorem}

\newcommand\bbR{\mathbb{R}}
\newcommand\bbN{\mathbb{N}}
\newcommand\bxi{\boldsymbol{\xi}}
\newcommand\bx{\boldsymbol{x}}
\newcommand\bv{\boldsymbol{v}}
\newcommand\bu{\boldsymbol{u}}
\newcommand\be{\boldsymbol{e}}
\newcommand\bC{\boldsymbol{C}}

\newcommand\bg{\boldsymbol{g}}
\newcommand\bh{\boldsymbol{h}}
\newcommand\bdeta{\boldsymbol{\eta}}

\newcommand\balpha{\boldsymbol{\alpha}}

\newcommand\dd{\,\mathrm{d}}
\newcommand\He{\mathit{He}}

\newcommand\bw{\boldsymbol{w}}

\numberwithin{equation}{section}

\newcommand\pd[2]{\dfrac{\partial {#1}}{\partial {#2}}}

\graphicspath{{images/}}

{\theoremstyle{remark} \newtheorem{remark}{Remark}}

\title{A Framework on Moment Model Reduction for Kinetic Equation}

\author{Zhenning Cai\thanks{Center for Computational Engineering
    Science, RWTH Aachen University, Aachen, Germany,
    email: {\tt cai@mathcces.rwth-aachen.de}.},~~
  Yuwei Fan\thanks{School of Mathematical Sciences, Peking University,
    Beijing, China, email: {\tt ywfan@pku.edu.cn}.},~~ Ruo
  Li\thanks{CAPT, LMAM \& School of Mathematical Sciences, Peking
    University, Beijing, China, email: {\tt rli@math.pku.edu.cn}.}}

\begin{document}
\maketitle
\begin{abstract}
  By a further investigation on the structure of the coefficient
  matrix of the globally hyperbolic regularized moment equations for
  Boltzmann equation in \cite{Fan}, we propose a uniform framework to
  carry out model reduction to general kinetic equations, to achieve
  certain moment system. With this framework, the underlying reason
  why the globally hyperbolic regularization in \cite{Fan} works is
  revealed. The even fascinating point is, with only routine
  calculation, existing models are represented and brand new models
  are discovered. Even if the study is restricted in the scope of the
  classical Grad's 13-moment system, new model with global
  hyperbolicity can be deduced.
\end{abstract}

\section{Introduction}
In 1949, the moment method was proposed by Grad \cite{Grad} for the
gas kinetic theory, and Grad's 13-moment theory is the most basic one
among the moment models beyond the Navier-Stokes theory. It has been
discovered in \cite{Muller} that the 1D reduction of Grad's 13-moment
equations is only hyperbolic around the equilibrium, and in
\cite{Grad13toR13}, it is found that in 3D case, the equilibrium is on
the boundary of the hyperbolicity region. Due to such a deficiency,
the application of the moment method is strongly limited. During a
long time, this remains a great obstacle to the improvement of the
moment method. However, through the effort of a large number of
researchers, the situation has become much more encouraging. Levermore
investigated the maximum entropy method and showed in \cite{Levermore}
that the moment system obtained by such a method owns global
hyperbolicity. Unfortunately, it is difficult to put it into
application owing to the lack of a simple analytical expression. Based
on Levermore's 14-moment closure, an affordable 14-moment closure is
proposed in \cite{McDonald} as an approximation, which also possesses
almost global hyperbolicity. Meanwhile, the discovery of globally
hyperbolic moment systems with large numbers of moments is also in
progress. In 1D case, two types of globally hyperbolic moment systems
are introduced in \cite{Fan} and \cite{Koellermeier} using different
strategies, and the approach in \cite{Koellermeier} successfully
brought us some insight on the globally hyperbolic regularization in
\cite{Fan}. The method in \cite{Fan} is extended to the
multidimensional case in two different versions \cite{Fan_new, ANRxx},
and some numerical results of the version in \cite{Fan_new} are
already carried out in \cite{Qiao}.

The research in this paper starts with a detailed study on the
globally hyperbolic regularization of Grad's moment method proposed in
\cite{Fan}. The regularization introduced in \cite{Fan} appears to be
the result of a direct demand of hyperbolicity with specified
characteristic speeds, and its underlying mechanism is not clarified,
which makes its extension to the multidimensional case mysterious. By
comparing the 1D regularized system with an adaptive discrete velocity
model, the coefficient matrix in the regularized system is properly
factorized and it is figured out how the system approximates the
Boltzmann equation. This new finding spontaneously leads to a new
deduction of the globally hyperbolic moment system, in which the time
derivative and the spatial derivative are discretized in the same way
so that one can achieve the global hyperbolicity by simply choosing a
real diagonalizable operator to discrete the multiplication of the
particle velocity. Due to this new deduction, the 1D globally
hyperbolic moment system is instantly obtained, and may no longer be
regarded as a regularization of the corresponding Grad's moment
system.

A significant advantage of this new deduction is its
extensibility. The whole procedure can be extended to the
multidimensional Boltzmann with more flexible ansatz on the
distribution function by only carrying out routine calculations. Based
on some simple assumptions, a framework for deriving ``moment
systems'' is established and we propose two conditions which ensure
the hyperbolicity. By further exploration, we point out that these two
conditions to hyperbolicity is hardly to be violated. This is on the
very contrary of the studies on the moment models in extensive
literatures, where the hyperbolicity is a subtle subject to be
achieved. Using this framework, both multidimensional models in
\cite{Fan_new, ANRxx} are reasonably deduced and a new globally
hyperbolic 13-moment system is discovered. Furthermore, this new
strategy can even be extended to a type of generic kinetic
equations. As an example, the radiative transfer equation is
investigated and it is found that the classic $M_N$ model is also
covered by the framework.

The rest of this paper is as follows: in Section \ref{sec:Boltzmann},
the Boltzmann equation and the moment method are briefly reviewed; in
Section \ref{sec:1D_struct}, the structure of the 1D globally
hyperbolic moment equations is studied and a new deduction is
proposed; in Section \ref{sec:uni_frmwk}, the deduction is generalized
to the multidimensional Boltzmann equation and later to the generic
kinetic equation; finally, some concluding remarks are made in Section
\ref{sec:conclusion}.


\section{Moment Method for Boltzmann Equation} \label{sec:Boltzmann}
In the gas kinetic theory, the movement of the gas molecules is
depicted from a statistical point of view. To be precise, it is
supposed that at any specific time $t$ and spatial point $\bx$, the
velocity of the gas molecule $\bxi$ is a random variable with
probability density function $p[t,\bx](\bxi)$, and the corresponding
distribution function is denoted as
\begin{equation}
f(t,\bx,\bxi) = n(t,\bx) \cdot p[t,\bx](\bxi),
  \qquad t \in \bbR^+, \quad \bx,\bxi \in \bbR^D,
\end{equation}
where $n(t,\bx)$ is the number density of the gas moleucles, and $D$
is the number of dimensions, which equals to $3$ in our physical
world. Denoting $\langle \phi \rangle = \displaystyle\int_{\bbR^D}
\phi \dd \bxi$ and assuming the gas to be an ideal gas, we have
\begin{equation} \label{eq:primitive}
\rho(t,\bx) = m_g \langle f \rangle, \quad
\rho \bu(t,\bx)  = m_g \langle \bxi f \rangle, \quad
\rho(t,\bx) R T(t,\bx)  = \frac{1}{D} m_g \langle |\bxi|^2 f \rangle,
\end{equation}
where the macroscopic functions $\rho$, $\bu$ and $T$ denotes the
density, velocity and temperature respectively, and the constant $R$
and $m_g$ denotes the gas constant and the mass of a single molecule.
By convention, we let $\theta(t,\bx) = R T(t,\bx)$.

The distribution function is governed by the Boltzmann equation, which
reads
\begin{equation} \label{eq:Boltzmann}
\frac{\partial f}{\partial t} + \bxi \cdot \nabla_{\bx} f =
  \mathcal{S}(f),
\end{equation}
where the right hand side $\mathcal{S}(f)$ is to model the effects of
the interaction among particles. In this paper, the concrete form of
$\mathcal{S}(f)$ is not concerned. We only assume here that if for some
$t_0$ and $\bx_0$,
\begin{equation} \label{eq:exp_quad}
f(t_0,\bx_0,\bxi) = \exp \left(
  A + \boldsymbol{B} \cdot \bxi + C |\bxi|^2
\right), \qquad
A \in \bbR, \quad \boldsymbol{B} \in \bbR^D, \quad C \in \bbR^-,
\end{equation}
then $\mathcal{S}(f)(t_0, \bx_0, \bxi) = 0$ (see e.g. \cite{Cowling}).
Using the relations in \eqref {eq:primitive}, we can write the
coefficients in \eqref{eq:exp_quad} as expressions of $\rho$, $\bu$
and $\theta$, and thus \eqref {eq:exp_quad} is reformulated by
\begin{equation}
f(t_0, \bx_0, \bxi) =
  \frac{\rho(t_0, \bx_0)}{m [2\pi \theta(t_0, \bx_0)]^{D/2}}
  \exp \left(
    -\frac{|\bxi - \bu(t_0, \bx_0)|^2}{2 \theta(t_0, \bx_0)}
  \right).
\end{equation}
This distribution function is actually the local equilibrium of the gas.

By the assumption on the operator $\mathcal{S}$, the unknown function
$f$ has to be close to the local equilibrium if $\mathcal{S}$ is
large. Based on this observation, Grad \cite{Grad} expanded the
distribution into the Hermite series. Here we follow \cite{NRxx} and
write the expansion as
\begin{equation} \label{eq:expansion}
f(t,\bx,\bxi) = \sum_{\balpha \in \bbN^D} 
  f_{\balpha}(t,\bx) \mathcal{H}_{\balpha}^{[\theta(t,\bx)]} \left(
    \frac{\bxi - \bu(t,\bx)}{\sqrt{\theta(t,\bx)}}
  \right).
\end{equation}
In the above expansion, $\balpha = (\alpha_1, \cdots, \alpha_D)$ is a
$D$-dimensional multi-index, and the basis functions
$\mathcal{H}_{\balpha}^{[\cdot]}(\cdot)$ is defined by
\begin{equation}
\mathcal{H}_{\balpha}^{[\theta]}(\bv) =
  m_g^{-1} \prod_{d=1}^D \frac{1}{\sqrt{2\pi}}
    \theta^{-\frac{\alpha_d+1}{2}}
    \He_{\alpha_d}(v_d) \exp \left(-\frac{v_d^2}{2} \right),
\quad \forall \bv = (v_1, \cdots, v_D) \in \bbR^D,
\end{equation}
where $\He_k(\cdot)$ is the Hermite polynomial of degree $k$, defined
by
\begin{equation}
\He_k(x) = (-1)^k \exp\left( \frac{x^2}{2} \right)
  \frac{\mathrm{d}^k}{\mathrm{d}x^k} \exp\left( -\frac{x^2}{2} \right).
\end{equation}
It is obvious that if $f_{\balpha} = 0$ for all $|\balpha| \geqslant
1$, then $f$ degenerates to the local equilibrium. Additionally, by
substituting \eqref {eq:expansion} into \eqref{eq:primitive}, one has
\begin{equation} \label{eq:restriction}
f_0 = \rho, \qquad f_{\balpha} \equiv 0 \text{~~if~} |\balpha| = 1,
  \qquad \sum_{|\balpha| = 1} f_{2\balpha} \equiv 0.
\end{equation}

Subsequently, the expansion \eqref{eq:expansion} is substituted into
the Boltzmann equation and the equations of $f_{\balpha}$ are deduced
by letting the coefficients of the expansion vanished. For simplicity,
here we only consider a model case $D = 1$, where $\balpha$ and $\bu$
are turned into scalars and we note them insteadly by $\alpha$ and
$u$. The resulted equations are
\begin{equation} \label{eq:full_mnt_eqs}
\begin{split}
& \frac{\partial \rho}{\partial t} +
  u \frac{\partial \rho}{\partial x} +
  \rho \frac{\partial u}{\partial x} = 0, \\
& \rho \frac{\partial u}{\partial t} +
  \theta \frac{\partial \rho}{\partial x} +
  \rho u \frac{\partial u}{\partial x} +
  \rho \frac{\partial \theta}{\partial x} = 0, \\
& \frac{1}{2} \rho \frac{\partial \theta}{\partial t} +
  \frac{1}{2} \rho u \frac{\partial \theta}{\partial x} +
  \rho u \frac{\partial u}{\partial x} +
  3 \frac{\partial f_3}{\partial x} = 0, \\
& \frac{\partial f_{\alpha}}{\partial t} -
  f_{\alpha-1} \frac{\theta}{\rho} \frac{\partial \rho}{\partial x} +
  (\alpha+1) f_{\alpha} \frac{\partial u}{\partial x} + \left(
    \frac{1}{2} \theta f_{\alpha-3} + \frac{\alpha-1}{2} f_{\alpha-1}
  \right) \frac{\partial \theta}{\partial x} \\
& \qquad -\frac{3}{\rho} f_{\alpha-2} \frac{\partial f_3}{\partial x} +
  \theta \frac{\partial f_{\alpha-1}}{\partial x} +
  u \frac{\partial f_\alpha}{\partial x} +
  (\alpha+1) \frac{\partial f_{\alpha+1}}{\partial x} =
  \mathcal{S}_{\alpha}, \quad \alpha \geqslant 3,
\end{split}
\end{equation}
where $\mathcal{S}_{\alpha}$ is obtained by expansiong of the
collision term which is out of our concern in this paper, and we
emphasize that $f_1 = f_2 = 0$ according to \eqref{eq:restriction}.

Obviously the equations in \eqref{eq:full_mnt_eqs} form a system with
infinite number of unknowns. In order to obtain a system with finite
number of equations, Grad suggested to choose a positive integer $M
\geqslant 2$, and then simply set $f_{\alpha}$ to be zero for all
$\alpha > M$, and discard all the equations containing $\partial_t
f_{\alpha}$ with $\alpha > M$. The resulting system is called Grad's
$(M+1)$-moment system.


\section{Globally Hyperbolic Regularized Moment System}
\label{sec:1D_struct}
In \cite{Fan}, the authors investigated the hyperbolicity\footnote{A
  first-order quasi-linear partial differential system
\begin{equation}
\frac{\partial \bw}{\partial t} +
  \sum_{k=1}^D {\bf A}_k(\bw) \frac{\partial \bw}{\partial x_k} =
  \boldsymbol{S}(\bw)
\end{equation}
is hyperbolic in some region $\Omega$ if and only if any linear
combinations of ${\bf A}_k(\bw)$, $k = 1,\cdots,D$ is real
diagonalizable for all $\bw \in \Omega$. } of 1D Grad's moment systems
for any $M$, and concluded that for $M \geqslant 3$, Grad's moment
system is only hyperbolic around the equilibrium. Moreover, a globally
hyperbolic moment regularization to Grad's system was then
proposed. However, it was not revealed in \cite{Fan} the intrinsic
structure of the regularized system, which makes the hyperbolicity
achieved therein a mystery. In this section, the underlying structure
of the globally hyperbolic regularization is discovered, and the
hyperbolicity is clarified as a natural deduction.

\subsection{Review of system in 1D case}
We first review the globally hyperbolic moment system in \cite{Fan}.
We write Grad's $(M+1)$-moment system as
\begin{equation}
\frac{\partial \bw}{\partial t} +
  {\bf A}(\bw) \frac{\partial \bw}{\partial x} =
  \boldsymbol{S}(\bw),
\end{equation}
where $\bw = (\rho, u, \theta, f_3, \cdots, f_M)^T$, and the matrix
${\bf A}(\bw)$ is given by
\begin{equation} \label{eq:A_M}
\scalebox{1.0}{\footnotesize %
\renewcommand{\arraystretch}{1.2}
$\begin{pmatrix}
  u & \rho & 0 & \hdotsfor{6} & 0\\
  \theta / \rho & u & 1 & 0 & \hdotsfor{5} & 0\\
  0 & 2 \theta & u & 6 / \rho & 0 & \hdotsfor{4} & 0\\
  0 & 4 f_3 & \rho \theta / 2 & u & 4 & 0 & \hdotsfor{3} & 0\\
  - \theta f_3 / \rho & 5 f_4 & 3 f_3 / 2 & \theta & u & 5 & 0 & \hdotsfor{2} & 0\\
  \hdotsfor{10} \\
  -\theta f_{M - 2} / \rho & M f_{M - 1} &
    \frac{1}{2}[(M - 2) f_{M - 2} + \theta f_{M - 4}] &
    -3 f_{M - 3} / \rho & 0 & \cdots & 0 &\theta & u & M\\
  -\theta f_{M-1} / \rho & (M+1) f_M &
    \frac{1}{2}[(M-1) f_{M - 1} + \theta f_{M - 3}] &
    -3 f_{M-2} / \rho & 0 & \hdotsfor{2} & 0 & \theta & u
\end{pmatrix}$.}
\end{equation}
In \cite{Fan}, it is shown that this matrix is real diagonalizable
only when $|f_M / (\rho\theta^{M/2})|$ and $|f_{M-1} /
(\rho\theta^{(M-1)/2})|$ are small enough. If we modify the matrix
${\bf A}(\bw)$ to a ``regularized coefficient matrix'' $\hat{\bf
A}(\bw)$ defined by
\begin{equation} \label{eq:hat_A_w}
\hat{\bf A}(\bw) := {\bf A}(\bw)
  -(M+1) f_M {\bf E}_{M+1,2} - \frac{M+1}{2} f_{M-1} {\bf E}_{M+1,3},
\end{equation}
then the ``regularized system''
\begin{equation} \label{eq:regularized}
\frac{\partial \bw}{\partial t} +
  \hat{\bf A}(\bw) \frac{\partial \bw}{\partial x} =
  \boldsymbol{S}(\bw)
\end{equation}
is globally hyperbolic. In equation \eqref{eq:hat_A_w}, ${\bf
E}_{p,q}$ stands for the $(M+1) \times (M+1)$ matrix $(E_{ij})$
with $E_{ij} = \delta_{ip} \delta_{jq}$. The characteristic polynomial
of $\hat{\bf A}(\bw)$ is
\begin{equation}
\det[\lambda {\bf I} - \hat{\bf A}(\bw)] =
  \theta^{\frac{M+1}{2}} \He_{M+1} \left(
    \frac{\lambda - u}{\sqrt{\theta}}
  \right).
\end{equation}
Therefore the characteristic speeds of \eqref{eq:regularized} are
\begin{equation}
u + c_j \sqrt{\theta}, \qquad j = 0, \cdots, M,
\end{equation}
where $c_0, \cdots, c_M$ are the distinct roots of the polynomial
$\He_{M+1}(x)$.

\subsection{Structure of the coefficient matrix} \label{sec:struct}
It has been discussed in \cite{Fan} that the system \eqref
{eq:regularized} is similar as a ``shifted and scaled discrete
velocity model'' whose discrete velocities are the characteristic
speeds of the system. Our investigation starts from this point. A
canonical discrete velocity model with characteristic speeds $\xi_0,
\cdots, \xi_M$ has the following form:
\begin{equation} \label{eq:dvm}
\frac{\partial f^{\mathrm{dvm}}_k}{\partial t} +
  \xi_k \frac{\partial f^{\mathrm{dvm}}_k}{\partial x} =
  S^{\mathrm{dvm}}_k, \qquad k = 0, \cdots, M.
\end{equation}
Here $f^{\mathrm{dvm}}_k(t,x)$ approximates $f(t,x,\xi_k)$, and
$S^{\mathrm{dvm}}_k(t,x)$ approximates $S(f)(t,x,\xi_k)$. For
convenience, we write \eqref{eq:dvm} as
\begin{equation} \label{eq:dvm_matrix}
\frac{\partial \boldsymbol{f}^{\mathrm{dvm}}}{\partial t} +
  {\bf\Xi} \frac{\partial \boldsymbol{f}^{\mathrm{dvm}}}{\partial x} =
  \boldsymbol{S}^{\mathrm{dvm}},
\end{equation}
where
\begin{equation}
\boldsymbol{f}^{\mathrm{dvm}} =
  (f^{\mathrm{dvm}}_0, \cdots, f^{\mathrm{dvm}}_M)^T, \quad
\boldsymbol{S}^{\mathrm{dvm}} = 
  (S^{\mathrm{dvm}}_0, \cdots, S^{\mathrm{dvm}}_M)^T, \quad
{\bf \Xi} = \mathrm{diag} \{ \xi_0, \cdots, \xi_M \}.
\end{equation}
In order to compare the regularized moment system
\eqref{eq:regularized} with \eqref{eq:dvm_matrix}, we first
factorize the matrix $\hat{\bf A}(\bw)$ by
\begin{equation}
\hat{\bf A}(\bw) =
  [{\bf L}(\bw)]^{-1} {\bf \Lambda}(\bw) {\bf L}(\bw), \quad
{\bf \Lambda}(\bw) = \mathrm{diag} \left\{
  u + c_0 \sqrt{\theta}, \cdots, u + c_M \sqrt{\theta}
\right\},
\end{equation}
and then substitute it into the system \eqref{eq:regularized} and
multiply both sides by ${\bf L}(\bw)$. This results in
\begin{equation} \label{eq:factorized}
{\bf L}(\bw) \frac{\partial \bw}{\partial t} +
  {\bf \Lambda}(\bw) {\bf L}(\bw) \frac{\partial \bw}{\partial x} =
  {\bf L}(\bw) {\boldsymbol{S}}(\bw).
\end{equation}
Thus it is clear that ${\bf L}(\bw) \partial_t \bw$ corresponds to
$\partial_t \boldsymbol{f}^{\mathrm{dvm}}$, and ${\bf L}(\bw)
\partial_x \bw$ corresponds to $\partial_x
\boldsymbol{f}^{\mathrm{dvm}}$. It can be verified for $M = 2$ that it
is impossible to find a vector function $\bg(\bw)$ such that
$\partial_{t,x} \bg(\bw) = {\bf L}(\bw) \partial_{t,x} \bw$.
Therefore, there does not exist a perfect correspondence between the
regularized moment equations and the discrete velocity model. However,
in the discrete velocity model, we may regard $\partial_{t,x}
\boldsymbol{f}^{\mathrm{dvm}}$ as the approximations of
\begin{displaymath}
\left( \frac{\partial f}{\partial (t,x)} (t,x,\xi_0), \cdots,
  \frac{\partial f}{\partial (t,x)} (t,x,\xi_M) \right)^T.
\end{displaymath}
Accordingly, ${\bf L}(\bw) \partial_{t,x} \bw$ is regarded as the
approximations of
\begin{displaymath}
\left(
  \frac{\partial f}{\partial (t,x)} (t,x,u + c_0 \sqrt{\theta}),
  \cdots, \frac{\partial f}{\partial (t,x)} (t,x,u + c_M \sqrt{\theta})
\right)^T.
\end{displaymath}
Since $u + c_j \sqrt{\theta}$, $j = 0, \cdots, M$, are Hermite-Gauss
quadrature points, such approximation is equivalent to approximating
the functions $\partial_{t,x} f$ by a finite Hermite expansion
\begin{equation} \label{eq:approx}
\frac{\partial f}{\partial (t,x)} \approx \sum_{\alpha = 0}^M
  D^{t,x}_{\alpha}(\bw) \mathcal{H}^{[\theta]}_{\alpha}
    \left( \frac{\xi - u}{\sqrt{\theta}} \right),
\end{equation}
where
\begin{equation} \label{eq:coef}
\left(
  D^{t,x}_0(\bw), \cdots, D^{t,x}_M(\bw)
\right)^T = {\bf T}^{[\theta]} {\bf L}(\bw) \partial_{t,x} \bw,
\end{equation}
and ${\bf T}^{[\theta]}$ is the transformation matrix from the values
on the quadrature points to the coefficients. Now the system \eqref
{eq:factorized} can be rewritten as
\begin{equation}
{\bf T}^{[\theta]} {\bf L}(\bw) \frac{\partial \bw}{\partial t} +
  \left[
    {\bf T}^{[\theta]} {\bf \Lambda}(\bw)
    \left( {\bf T}^{[\theta]} \right)^{-1}
  \right] {\bf T}^{[\theta]} {\bf L}(\bw)
    \frac{\partial \bw}{\partial x} =
{\bf T}^{[\theta]} {\bf L}(\bw) \boldsymbol{S}(\bw).
\end{equation}
Below we are going to calculate the matrices ${\bf T}^{[\theta]} {\bf
\Lambda}(\bw) \left( {\bf T}^{[\theta]} \right)^{-1}$ and ${\bf
T}^{[\theta]} {\bf L}(\bw)$.

Let $\bh = (h_0, \cdots, h_M)^T$ be a $(M+1)$-dimensional vector and
define
\begin{equation} \label{eq:H}
H(v) = \sum_{\alpha=0}^M h_{\alpha}\mathcal{H}^{[\theta]}_{\alpha}(v).
\end{equation}
According to the definition of ${\bf T}^{[\theta]}$, we have
\begin{equation}
\big( {\bf T}^{[\theta]} \big)^{-1} \bh =
  \left( H(c_0), \cdots, H(c_M) \right)^T,
\end{equation}
and then
\begin{equation}
{\bf \Lambda}(\bw) \big( {\bf T}^{[\theta]} \big)^{-1} \bh =
  \left( (u + c_0 \sqrt{\theta}) H(c_0), \cdots,
    (u + c_M \sqrt{\theta}) H(c_M) \right)^T.
\end{equation}
If we define $\tilde{\bh} = (\tilde{h}_0, \cdots, \tilde{h}_M)^T =
{\bf T}^{[\theta]} {\bf \Lambda}(\bw) \left( {\bf T}^{[\theta]}
\right)^{-1} \bh$, and
\begin{equation} \label{eq:tilde_H}
\tilde{H}(v) = \sum_{\alpha=0}^M
  \tilde{h}_{\alpha}\mathcal{H}^{[\theta]}_{\alpha}(v),
\end{equation}
then $\tilde{\bh}$ is the unique vector such that
\begin{equation} \label{eq:tilde_H_eq_xi_H}
\left( \tilde{H}(c_0), \cdots, \tilde{H}(c_M) \right)^T =
  \left( (u + c_0 \sqrt{\theta}) H(c_0), \cdots,
    (u + c_M \sqrt{\theta}) H(c_M) \right)^T.
\end{equation}
By the recursion relation of the Hermite polynomials, it is obtained
that
\begin{equation}
\begin{split}
(u + v \sqrt{\theta}) H(v) &= \sum_{\alpha=0}^M h_{\alpha}
\left[
  \theta \mathcal{H}^{[\theta]}_{\alpha+1}(v) +
  u \mathcal{H}^{[\theta]}_{\alpha}(v) +
  \alpha \mathcal{H}^{[\theta]}_{\alpha-1}(v)
\right] \\
& = \sum_{\alpha=0}^M \left[
  \theta h_{\alpha-1} + u h_{\alpha} + (\alpha + 1) h_{\alpha + 1}
\right] \mathcal{H}^{[\theta]}_{\alpha} (v) +
\theta h_M \mathcal{H}^{[\theta]}_{M+1}(v),
\end{split}
\end{equation}
where $\mathcal{H}_{-1}^{[\theta]}(v)$, $h_{-1}$ and $h_{M+1}$ are
taken as zeros. Since $\mathcal{H}_{M+1}^{[\theta]}(c_j) = 0$, $j = 0,
\cdots, M$, we have
\begin{equation} \label{eq:xi_H}
(u + c_j \sqrt{\theta}) H(c_j) = \sum_{\alpha=0}^M \left[
  \theta h_{\alpha-1} + u h_{\alpha} + (\alpha + 1) h_{\alpha + 1}
\right] \mathcal{H}^{[\theta]}_{\alpha} (c_j), \qquad
  j = 0, \cdots, M.
\end{equation}
Collecting \eqref{eq:tilde_H}, \eqref{eq:tilde_H_eq_xi_H} and
\eqref{eq:xi_H} together, one instantly figures out that
\begin{equation}
\tilde{h}_{\alpha} =
  \theta h_{\alpha-1} + u h_{\alpha} + (\alpha + 1) h_{\alpha+1},
\qquad \alpha = 0, \cdots, M.
\end{equation}
Thus the matrix ${\bf T}^{[\theta]} {\bf \Lambda}(\bw) \left( {\bf
    T}^{[\theta]} \right)^{-1}$ is obviously a tridiagonal matrix. The
diagonal entries are all $u$, the subdiagonal entries are all
$\theta$, and the superdiagonal entries are $1,2,\cdots,M$. This
matrix is denoted as ${\bf M}(u,\theta)$ below. The analysis reveals
that ${\bf M}(u,\theta)$ actually acts as an operator which multiply a
function like \eqref{eq:H} by $\xi$, and then drop out the term with
the highest degree.

In order to caluclate the matrix ${\bf T}^{[\theta]} {\bf L}(\bw)$, we
first consider small values of $M$, and calculate it directly using
some computer algebra system. Then, by simple induction, we find that
the matrix ${\bf T}^{[\theta]} {\bf L}(\bw)$ has the following general
form:
\begin{equation} \label{eq:block_D}
{\bf D}(\bw) = \begin{pmatrix}
  {\bf D}_{11}(\rho) & {\bf 0} \\
  {\bf D}_{21}(\bw) & {\bf I}_{M-3}
\end{pmatrix},
\end{equation}
where ${\bf D}_{11} = \mathrm{diag} \{1, \rho, \rho/2, 1\}$, and
\begin{equation}
{\bf D}_{21} = \begin{pmatrix}
0 & f_3 & 0 & 0 \\
0 & f_4 & \frac{1}{2} f_3 & 0 \\
0 & f_5 & \frac{1}{2} f_4 & 0 \\
\vdots & \vdots & \vdots & \vdots \\
0 & f_{M-1} & \frac{1}{2} f_{M-2} & 0
\end{pmatrix}.
\end{equation}
One can validate ${\bf D}(\bw) = {\bf T}^{[\theta]} {\bf L}(\bw)$ by
verifying ${\bf D}(\bw) {\bf A}(\bw) = {\bf M}(u,\theta) {\bf D}(\bw)$
with direct calculation. And finally the regularized moment system
\eqref{eq:regularized} is rewritten as
\begin{equation}
{\bf D}(\bw) \frac{\partial \bw}{\partial t} +
  {\bf M}(\bw) {\bf D}(\bw) \frac{\partial \bw}{\partial x} = 
  {\bf D}(\bw) \boldsymbol{S}(\bw).
\end{equation}

In order to finish the comparison between the regularized moment
system and the discrete velocity model, we turn back to the
approximation \eqref{eq:approx} and try to find out how the right hand
side of \eqref{eq:approx} approximates the derivatives on the left
hand side. For convenience, we consider only the spatial derivative.
Taking spatial derivatives on both sides of \eqref{eq:expansion} with
$D = 1$, we have the following expansion of $\partial_x f$:
\begin{equation}
\frac{\partial f}{\partial x} = \sum_{\alpha = 0}^{+\infty} \left(
  \frac{\partial f_{\alpha}}{\partial x} +
  \frac{\partial u}{\partial x} f_{\alpha - 1} +
  \frac{1}{2} \frac{\partial \theta}{\partial x} f_{\alpha - 2}
\right) \mathcal{H}^{[\theta]}_{\alpha} \left(
  \frac{\xi - u}{\sqrt{\theta}}
\right),
\end{equation}
where $f_{-1} = f_{-2} = 0$ (See \cite{NRxx_new} for the detailed
procedure). We calculate the coefficients $D_{\alpha}^x(\bw)$ using
\eqref{eq:block_D} and \eqref{eq:coef}, and find out that the right
hand side of $\eqref{eq:approx}$ is
\begin{equation}
\sum_{\alpha = 0}^M
  D^x_{\alpha}(\bw) \mathcal{H}^{[\theta]}_{\alpha}
    \left( \frac{\xi - u}{\sqrt{\theta}} \right) =
\sum_{\alpha = 0}^M \left(
  \frac{\partial f_{\alpha}}{\partial x} +
  \frac{\partial u}{\partial x} f_{\alpha - 1} +
  \frac{1}{2} \frac{\partial \theta}{\partial x} f_{\alpha - 2}
\right) \mathcal{H}^{[\theta]}_{\alpha} \left(
  \frac{\xi - u}{\sqrt{\theta}}
\right).
\end{equation}
Now it is clear that the approximation \eqref{eq:approx} is actually a
direct truncation in the Hermite expansion.

As a summary, through a comparison with the discrete velocity model,
we figure out the underlying mechanism of the regularized moment
system, which evidently makes the reasonability and credibility of
this new model more pronounced. And consequentially, this motivates us
to propose the approach below to deduce the globally hyperbolic moment
system directly.

\subsection{New deduction procedure}
\label{sec:new_deduction}
In \cite{Fan}, the system is obtained by first deriving Grad's moment
system and then applying the regularization. Based on the discussion
above, we now may deduce the regularized moment system directly from
the Boltzmann equation. The procedure is as follows:
\begin{enumerate}
\item Expand the distribution function into the Hermite series:
\begin{equation} \label{eq:1D_expansion}
f(t,x,\xi) = \sum_{\alpha=0}^{+\infty} f_{\alpha}(t,x)
  \mathcal{H}^{[\theta]}_{\alpha} \left(
    \frac{\xi - u}{\sqrt{\theta}}
  \right).
\end{equation}
\item Calculate the time and spatial derivatives and the collision
term:
\begin{subequations} \label{eq:full_expansion}
\begin{align}
\frac{\partial f}{\partial t} &= \sum_{\alpha=0}^{+\infty}
  G^t_{\alpha}(t,x) \mathcal{H}^{[\theta]}_{\alpha} \left(
    \frac{\xi - u(t,x)}{\sqrt{\theta(t,x)}}
  \right), \\
\frac{\partial f}{\partial x} &= \sum_{\alpha=0}^{+\infty}
  G^x_{\alpha}(t,x) \mathcal{H}^{[\theta]}_{\alpha} \left(
    \frac{\xi - u(t,x)}{\sqrt{\theta(t,x)}}
  \right), \\
S(f) &= \sum_{\alpha=0}^{+\infty}
  Q_{\alpha}(t,x) \mathcal{H}^{[\theta]}_{\alpha} \left(
    \frac{\xi - u(t,x)}{\sqrt{\theta(t,x)}}
  \right),
\end{align}
\end{subequations}
where
\begin{equation}
G^t_{\alpha} = \frac{\partial f_{\alpha}}{\partial t} +
  \frac{\partial u}{\partial t} f_{\alpha - 1} +
  \frac{1}{2} \frac{\partial \theta}{\partial t} f_{\alpha - 2}, \quad
G^x_{\alpha} = \frac{\partial f_{\alpha}}{\partial x} +
  \frac{\partial u}{\partial x} f_{\alpha - 1} +
  \frac{1}{2} \frac{\partial \theta}{\partial x} f_{\alpha - 2},
\end{equation}
and $Q_{\alpha}(t,x)$ depends on the collision model.
\item \label{step:first_cut_off} For the positive integer $M \geqslant
  2$, apply a truncation on \eqref{eq:full_expansion}. The truncation
  drops off all $f_{\alpha}$ with $\alpha > M$, and discards all terms
  with $\alpha > M$ in the summations in \eqref{eq:full_expansion}. We
  write the results as
\begin{subequations}
\begin{align}
\label{eq:time_derivative}
\frac{\partial f}{\partial t} &\approx \sum_{\alpha=0}^M
  G^t_{\alpha}(t,x) \mathcal{H}^{[\theta]}_{\alpha} \left(
    \frac{\xi - u(t,x)}{\sqrt{\theta(t,x)}}
  \right), \\
\label{eq:spatial_derivative}
\frac{\partial f}{\partial x} &\approx \sum_{\alpha=0}^M
  G^x_{\alpha}(t,x) \mathcal{H}^{[\theta]}_{\alpha} \left(
    \frac{\xi - u(t,x)}{\sqrt{\theta(t,x)}}
  \right), \\
\label{eq:collision}
S(f) & \approx \sum_{\alpha=0}^M
  \tilde{Q}_{\alpha}(t,x) \mathcal{H}^{[\theta]}_{\alpha} \left(
    \frac{\xi - u(t,x)}{\sqrt{\theta(t,x)}}
  \right).
\end{align}
\end{subequations}
\item Calculate the convection term $\xi \partial_x f$ from \eqref
{eq:spatial_derivative}:
\begin{equation} \label{eq:convection}
\xi \frac{\partial f}{\partial x} \approx
  \xi \sum_{\alpha=0}^M
  G^x_{\alpha}(t,x) \mathcal{H}^{[\theta]}_{\alpha} \left(
    \frac{\xi - u(t,x)}{\sqrt{\theta(t,x)}}
  \right) =
  \sum_{\alpha=0}^{M+1}
    J_{\alpha}(t,x) \mathcal{H}^{[\theta]}_{\alpha} \left(
      \frac{\xi - u(t,x)}{\sqrt{\theta(t,x)}}
    \right),
\end{equation}
where $J_{\alpha} = \theta G^x_{\alpha - 1} + u G^x_{\alpha} + (\alpha
+ 1) G^x_{\alpha + 1}$.
\item \label{step:second_cut_off} Truncate \eqref{eq:convection} again:
\begin{equation}
\xi \frac{\partial f}{\partial x} \approx \sum_{\alpha=0}^{M}
  J_{\alpha}(t,x) \mathcal{H}^{[\theta]}_{\alpha} \left(
    \frac{\xi - u(t,x)}{\sqrt{\theta(t,x)}}
  \right),
\end{equation}
\item Substitute \eqref{eq:time_derivative}, \eqref{eq:collision} and
  \eqref{eq:convection} into the Boltzmann equation and extract the
  coefficients for the same basis functions. The result obtained
\begin{equation}
G^t_{\alpha} + J_{\alpha} = \tilde{Q}_{\alpha},
  \quad 0 \leqslant \alpha \leqslant M
\end{equation}
is the hyperbolic $(M+1)$-moment system.
\end{enumerate}
Comparing with this procedure with the analysis in section
\ref{sec:struct}, it is clear that the vectors ${\bf D}(\bw) \bw_t$,
${\bf D}(\bw) \bw_x$ and ${\bf D}(\bw) \boldsymbol{S}(\bw)$ are
constructed by step \ref{step:first_cut_off}, and the vector ${\bf
  M}(u,\theta) {\bf D}(\bw) \bw_x$ is constructed by step \ref
{step:second_cut_off}.

By this procedure, the globally hyperbolic moment systems may be
regarded directly as an approximation of Boltzmann equation, instead
of the consequence of the regularization of Grad's moment systems. The
key point during this deduction is that the truncation has to be
applied instantly after every single operation, including both taking
the derivatives of $x$ or $t$, and the multiplication by velocity.


\section{A Uniform Framework for Generic Kinetic Equation}
\label{sec:uni_frmwk}
In the last section, a brand new approach to deduce the globally
hyperbolic moment system for Boltzmann equation is proposed. The
approach is to be extended to a uniform framework hereafter. With this
framework, the hyperbolicity of the obtained system can be intuitively
perceived and no intricate manipulation of the coefficient matrices
is needed. Based on the deduction for 1-dimensional problem in the
last section, we first extend this to the multi-dimensional Boltzmann
equation. And actually, this procedure can be made even more
general. We will reveal this fact step by step in this section.

\subsection{Multi-dimensional Boltzmann equation}
\label{sec:MD_Boltzmann}
Let us go back to the original Boltzmann equation \eqref
{eq:Boltzmann}. In order to include more general models, we consider a
general expansion of the distribution function:
\begin{equation} \label{eq:general_expansion}
f(t,\bx,\bxi) = \sum_{i = 0}^{+\infty} F_i(t,\bx)
  \varphi^{[\bdeta(t,\bx)]}_i(\bxi).
\end{equation}
Note that here we allow the basis functions $\varphi_i$ to be
dependent on some unknown vector function $\bdeta(t,\bx) =
(\eta_1(t,\bx), \cdots, \eta_n(t,\bx))$, and we suppose for any fixed
$\bdeta$, the basis functions $\{\varphi^{[\bdeta]}_i(\bxi)\}_{i=1}%
^{\infty}$ belong to some Hilbert space $\mathbb{H}^{[\bdeta]}$ with
real inner product $\langle \cdot, \cdot \rangle^{[\bdeta]}$. Due to
the parameter $\bdeta$ in the basis function, some constraints
\begin{equation} \label{eq:constraints}
r_j(\eta_1, \cdots, \eta_n, F_0, F_1, \cdots, F_m) = 0,
  \qquad j = 1,\cdots, n
\end{equation}
have to be imposed on the coefficients, and \eqref{eq:constraints} is
a non-degenerate algebraic system, which is provided by that its
Jacobian matrix
\[
\begin{pmatrix}
\pd{r_1}{\eta_1} & \cdots & \pd{r_1}{\eta_n} & \pd{r_1}{F_0} & \cdots & \pd{r_1}{F_m} \\[10pt]
\pd{r_2}{\eta_1} & \cdots & \pd{r_2}{\eta_n} & \pd{r_2}{F_0} & \cdots & \pd{r_2}{F_m} \\
\vdots & \ddots & \vdots & \vdots & \ddots & \vdots \\
\pd{r_n}{\eta_1} & \cdots & \pd{r_n}{\eta_n} & \pd{r_n}{F_0} & \cdots & \pd{r_n}{F_m}
\end{pmatrix},
\]
has full row rank. For example, in expansion \eqref{eq:1D_expansion},
we have $\bdeta = (u, \theta)$, and the constraints are $f_1 = f_2 =
0$, $m = 2$.

We suppose the functions $r_1, \cdots, r_n$ are continuously
differentiable. Owing to the non-degeneracy of \eqref{eq:constraints}
and the implicit function theorem, around any solution of \eqref
{eq:constraints}, there exists one $n$-dimensional subvector of
$(\eta_1, \cdots, \eta_n, F_0, \cdots, F_m)^T$ such that it can be
written as a function of else variables. Thus, in order to construct
a system with $(M + 1)$ equations, we have to consider $(M + n + 1)$
functions $\eta_1, \cdots, \eta_n, F_0, F_1, \cdots, F_M$, and then
``close'' the distribution function by writing all $F_i$, $i > M$ as
algebraic functions of $\eta_1, \cdots, \eta_n, F_0, F_1, \cdots,
F_M$. In order that \eqref{eq:constraints} remains non-degenerate, we
require $M \geqslant \max\{m,n\}$. Thus, the implicit function theorem
can be applied to eliminate $n$ unknowns, and we denote by $\bw$ the
vector whose components are the remaining $(M+1)$ functions, and the
distribution function $f$ is actually approximated by
\begin{equation} \label{eq:step1}
f(t,\bx,\bxi) \approx \tilde{f}(t,\bx,\bxi) =
  \sum_{i=0}^M F_i(t,\bx) \varphi_i^{[\bdeta]}(\bxi)
  + \sum_{i=M+1}^{+\infty} F_i(\bw(t,\bx)) \varphi_i^{[\bdeta]}(\bxi).
\end{equation}

Now we mimic the procedure in section \ref{sec:new_deduction} to
derive the equations for $\bw$. We first calculate the time and
spatial derivatives of $\tilde{f}$, together with the collision term:
\begin{equation} \label{eq:step2}
\begin{split}
& \frac{\partial \tilde{f}}{\partial (t,x_1, \cdots, x_D)} =
\sum_{i=0}^{+\infty} \left[
  G_i^t(t,\bx), G_i^{x_1}(t,\bx), \cdots, G_i^{x_D}(t,\bx)
\right] \varphi_i^{[\bdeta(t,\bx)]}(\bxi), \\
& \mathcal{S}(\tilde{f}) = \sum_{i=0}^{+\infty}
  Q_i(t,\bx) \varphi_i^{[\bdeta(t,\bx)]}(\bxi),
\end{split}
\end{equation}
where $G_i^t$ and $G_i^x$ are formulated as expressions of $\eta_1,
\cdots, \eta_n, F_0, F_1, \cdots, F_M$ and their derivatives. In order
to define the ``truncation'', we denote by $\mathbb{S}_M^{[\bdeta]}$
the linear subspace of $\mathbb{H}^{[\bdeta]}$ spanned by
$\{\varphi_0^{[\bdeta]}, \cdots, \varphi_M^{[\bdeta]}\}$, and then
approximate \eqref{eq:step2} by projecting it into this finite
dimensional space $\mathbb{S}_M^{[\bdeta]}$:
\begin{subequations} \label{eq:step3}
\begin{align}
\label{eq:lhs}
& \frac{\partial \tilde{f}}{\partial (t,x_1, \cdots, x_D)} \approx
\sum_{i=0}^M \left[
  \tilde{G}_i^t(t,\bx), \tilde{G}_i^{x_1}(t,\bx),
    \cdots, \tilde{G}_i^{x_D}(t,\bx)
\right] \varphi_i^{[\bdeta(t,\bx)]}(\bxi), \\
& \mathcal{S}(\tilde{f}) \approx \sum_{i=0}^M
  \tilde{Q}_i(t,\bx) \varphi_i^{[\bdeta(t,\bx)]}(\bxi).
\end{align}
\end{subequations}
If $\varphi_0^{[\bdeta]}, \varphi_1^{[\bdeta]}, \cdots, $ are
orthogonal to each other, then \eqref{eq:step3} is a direct truncation
of \eqref{eq:step2}, and otherwise, the coefficients in
\eqref{eq:step3} has the following form:
\begin{equation} \label{eq:projection}
\begin{split}
\tilde{G}_i^s(t,\bx) = \sum_{j=0}^{+\infty}
  a_{ij}(\bdeta(t,\bx)) G_j^s(t,\bx), \quad
\tilde{Q}_i(t,\bx) = \sum_{j=0}^{+\infty}
  a_{ij}(\bdeta(t,\bx)) Q_j(t,\bx), \quad \\
\qquad s = t,x_1,\cdots,x_D, \quad i=0,\cdots,M,
\end{split}
\end{equation}
Here $a_{ij}$ are only algebraic functions of $\bdeta$, and they do
not contain any derivatives of $\bdeta$. Let $\tilde{\bg}^s =
(\tilde{G}_0^s, \cdots, \tilde{G}_M^s)^T$, $s = t,x_1,\cdots,x_D$.
Using the chain rule and denoting all the unknown functions by $\bw$,
the vector $\tilde{\bg}^s$ has the form
\begin{equation} \label{eq:tilde_g}
\tilde{\bg}^s = {\bf D}(\bw) \frac{\partial \bw}{\partial s},
  \qquad s = t,x_1, \cdots, x_D.
\end{equation}
Based on \eqref{eq:lhs}, we approximate $\xi_k \partial_{x_k} f$ by
\begin{equation} \label{eq:step4}
\xi_k \frac{\partial \tilde{f}}{\partial x_k} \approx
  \sum_{i=0}^{+\infty} J_{k,i}(t,\bx) \varphi_i^{[\bdeta(t,\bx)]}(\bxi),
\qquad k = 1,\cdots,D,
\end{equation}
and then project it into the space $\mathbb{S}_M^{[\bdeta]}$ similar
as \eqref{eq:step3}:
\begin{equation} \label{eq:step5}
\xi_k \frac{\partial \tilde{f}}{\partial x_k} \approx
  \sum_{i=0}^{M} \tilde{J}_{k,i}(t,\bx)
    \varphi_i^{[\bdeta(t,\bx)]}(\bxi),
\qquad k = 1,\cdots,D,
\end{equation}
Finally, we put \eqref{eq:step5} and \eqref{eq:step3} into the
Boltzmann equation, and make the following system by vanishing the
coefficients of terms in the expansion:
\begin{equation} \label{eq:step6}
\tilde{G}_i^t(t,\bx) + \sum_{k=1}^D \tilde{J}_{k,i}(t,x) =
  \tilde{Q}_i(t,\bx), \qquad i = 0, \cdots, M.
\end{equation}
Similar as \eqref{eq:tilde_g}, we let $\boldsymbol{l}_k =
(\tilde{J}_{k,0},\cdots, \tilde{J}_{k,M})^T$ for $k = 1,\cdots,D$, and
then $\boldsymbol{l}_k$ has the form of
\begin{equation} \label{eq:l_k}
\boldsymbol{l}_k = {\bf M}_k(\bw) \tilde{\bg}^{x_k},
  \qquad k = 1,\cdots,D.
\end{equation}
Thus, the system \eqref{eq:step6} is formulated as
\begin{equation} \label{eq:abs_system}
{\bf D}(\bw) \frac{\partial \bw}{\partial t} + \sum_{k=1}^D
  {\bf M}_k(\bw) {\bf D}(\bw) \frac{\partial \bw}{\partial x_k}
= \boldsymbol{q}(\bw),
\end{equation}
where $\boldsymbol{q} = (\tilde{Q}_0, \cdots, \tilde{Q}_M)^T$. It is
easy to find that the system \eqref{eq:abs_system} is hyperbolic if
\begin{enumerate}
\item ${\bf D}(\bw)$ is invertible;
\item any linear combination of ${\bf M}_k(\bw)$ is real
  diagonalizable.
\end{enumerate}
We will point out later that these two conditions are fulfiled in very
extensive configurations.

In the above procedure, the vector of unknown functions $\bw$ is
defined only locally, as the result of the implicit function
theorem. However, in a lot of cases (e.g. the case in section
\ref{sec:1D_struct}), $\bw$ can be defined in the large, or even
globally. Thus the system \eqref{eq:abs_system} becomes globally
hyperbolic only if $\bw$ is properly chosen that ${\bf D}(\bw)$ is
invertible. Below, we give two examples to show how this abstract
procedure works.

\subsubsection{Example 1: 13-moment system}
\label{sec:13m}
When the moment method was proposed by Grad \cite{Grad}, he
immediately deduced a 13-moment system as its first application.
However, this system suffers a serious problem with its hyperbolicity.
It has been revealed in \cite{Grad13toR13} that the hyperbolicity of
this system cannot be ensured even around the equilibrium. To our
current knowledge, there has no a 13-moment system with global
hyperbolicity yet. Here we are going to deduce a globally hyperbolic
13-moment system using the framework introduced above.

In and only in this part, we will adopt the Einstein summation
convention to simplify the notations. Since 13-moment system is based
on 3D spatial and velocity spaces, the subscripts run over from $1$ to
$3$. Following Grad's work, we first approximate the distribution
function as
\begin{equation} \label{eq:13m_expansion}
\begin{split}
& f(t,\bx,\bxi) \approx f_{|13}(t,\bx,\bxi) =
  \rho(t,\bx) w^{[\theta(t,\bx)]}(\bC(t,\bx)) +
  \kappa_i(t,\bx) \mathscr{H}^{[\theta(t,\bx)]}_i (\bC(t,\bx)) \\
& \hspace{2.5cm} +
  \frac{1}{2} \rho(t,\bx) \theta_{\langle ij \rangle}(t,\bx)
    \mathscr{H}^{[\theta(t,\bx)]}_{ij}(\bC(t,\bx)) +
  \frac{1}{5} q_j(t,\bx)
    \mathscr{H}^{[\theta(t,\bx)]}_{iij}(\bC(t,\bx)),
\end{split}
\end{equation}
where $\theta_{ij}$ is a symmetric tensor, and
\begin{equation}
\theta_{\langle ij \rangle} = \theta_{ij} - \delta_{ij} \theta_{kk}/3,
  \qquad \bC = \bxi - \bu.
\end{equation}
The basis functions are defined as
\begin{equation}
w^{[\theta]}(\bC) = \frac{1}{m_g (2\pi\theta)^{3/2}}
  \exp \left( -\frac{C_k C_k}{2\theta} \right), \qquad
\mathscr{H}^{[\theta]}_{i_1 \cdots i_n}(\bC) =
  (-1)^n \frac{\partial^n w^{[\theta]}(\bC)}
    {\partial C_{i_1} \cdots \partial C_{i_n}}.
\end{equation}
In this case, $\bdeta = (u_1, u_2, u_3, \theta)$, and the constraints
\eqref{eq:constraints} are
\begin{equation}
\kappa_1(t,\bx) = \kappa_2(t,\bx) = \kappa_3(t,\bx) = 0,
  \qquad \theta_{kk}(t,\bx) = 3\theta(t,\bx).
\end{equation}
The inner product is defined by
\begin{equation} \label{eq:13m_ip}
\langle g_1, g_2 \rangle^{[\bdeta]} = \int_{\mathbb{R}^D}
  \frac{g_1(\bxi) g_2(\bxi)}{w^{[\theta]}(\bxi - \bu)} \dd \bxi.
\end{equation}
The vector $\bw$ of unknowns is chosen as
\begin{equation}
\bw = (\rho, u_1, u_2, u_3, \theta_{11}, \theta_{22}, \theta_{33},
  \theta_{12}, \theta_{13}, \theta_{23}, q_1, q_2, q_3)^T.
\end{equation}

Taking derivatives on both sides of \eqref{eq:13m_expansion}, we have
\begin{equation} \label{eq:13m_diff}
\begin{split}
\frac{\partial f_{|13}}{\partial s} & =
  \frac{\partial \rho}{\partial s} w^{[\theta]}(\bC) +
  \rho \frac{\partial u_i}{\partial s} \mathscr{H}^{[\theta]}_i(\bC) +
  \frac{1}{2}\rho \frac{\partial\theta}{\partial t}
    \mathscr{H}^{[\theta]}_{jj}(\bC) +
  \frac{1}{2}\rho \theta_{\langle ij \rangle}
    \frac{\partial u_k}{\partial s} \mathscr{H}^{[\theta]}_{ijk}(\bC) \\
& \quad +
  \frac{1}{5} \frac{\partial q_j}{\partial s}
    \mathscr{H}^{[\theta]}_{iij}(\bC) +
  \frac{1}{4} \rho \theta_{\langle ij \rangle}
    \frac{\partial \theta}{\partial s}
    \mathscr{H}^{[\theta]}_{ijkk}(\bC) +
  \frac{1}{5} q_j \frac{\partial u_k}{\partial s}
    \mathscr{H}^{[\theta]}_{iijk}(\bC) +
  \frac{1}{10} q_j \frac{\partial \theta}{\partial s}
    \mathscr{H}^{[\theta]}_{iijkk}(\bC),
\end{split}
\end{equation}
for $s = t,x_1,x_2,x_3$.  Let $\boldsymbol{B}$ be the vector whose
components are the following $13$ basis functions:
\begin{equation}
\begin{split}
\boldsymbol{B} = \Big( & w^{[\theta]}(\bC), \\
& \mathscr{H}^{[\theta]}_1(\bC),
  \mathscr{H}^{[\theta]}_2(\bC),
  \mathscr{H}^{[\theta]}_3(\bC), \\
& \mathscr{H}^{[\theta]}_{11}(\bC), \mathscr{H}^{[\theta]}_{22}(\bC),
  \mathscr{H}^{[\theta]}_{33}(\bC), \mathscr{H}^{[\theta]}_{12}(\bC),
  \mathscr{H}^{[\theta]}_{13}(\bC), \mathscr{H}^{[\theta]}_{23}(\bC), \\
& \mathscr{H}^{[\theta]}_{ii1}(\bC),
  \mathscr{H}^{[\theta]}_{ii2}(\bC),
  \mathscr{H}^{[\theta]}_{ii3}(\bC)
\Big)^T,
\end{split}
\end{equation}
and denote by $\mathbb{S}_{13}^{[\bdeta]}$ the linear space spanned by
these basis functions. Then we can approximate \eqref{eq:13m_diff} by
projecting it into the $13$-dimensional space
$\mathbb{S}_{13}^{[\bdeta]}$. The result is
\begin{equation} \label{eq:13m_cut_off}
\begin{split}
\frac{\partial f_{|13}}{\partial s} & \approx
  \frac{\partial \rho}{\partial s} w^{[\theta]}(\bC) +
  \rho \frac{\partial u_i}{\partial s} \mathscr{H}^{[\theta]}_i(\bC) \\
& \qquad + \frac{1}{2} \left(
    \rho \frac{\partial \theta_{ij}}{\partial s} +
    \frac{\partial \rho}{\partial s} \theta_{\langle ij \rangle}
  \right) \mathscr{H}^{[\theta]}_{ij}(\bC) +
  \frac{1}{5} \left(
    \rho \theta_{\langle jk \rangle} \frac{\partial u_k}{\partial s} +
    \frac{\partial q_j}{\partial s}
  \right) \mathscr{H}^{[\theta]}_{iij}(\bC).
\end{split}
\end{equation}
The above approximation has the form
\begin{equation}
\frac{\partial f_{|13}}{\partial s} \approx
  \boldsymbol{B}^T {\bf D}(\bw) \frac{\partial \bw}{\partial s},
\end{equation}
where ${\bf D}(\bw)$ is the very matrix in \eqref{eq:tilde_g}. It is
found from \eqref{eq:13m_cut_off} that ${\bf D}(\bw)$ is a lower
triangular square matrix with its diagonal entries as
\begin{displaymath}
1, \rho, \rho, \rho,
\frac{1}{2} \rho, \frac{1}{2} \rho, \frac{1}{2} \rho,
\rho, \rho, \rho, \frac{1}{5}, \frac{1}{5}, \frac{1}{5}.
\end{displaymath}
Therefore ${\bf D}(\bw)$ is invertible.

Now we calculate the convection term $\xi_k \partial_{x_k} f$ from
\eqref{eq:13m_cut_off}:
\begin{equation} \label{eq:13m_convection}
\begin{split}
\xi_k \frac{\partial f_{|13}}{\partial x_k} & \approx \left(
  u_k \frac{\partial \rho}{\partial x_k} +
  \rho \frac{\partial u_k}{\partial x_k}
\right) w^{[\theta]}(\bC) + \left(
  \rho u_k \frac{\partial u_i}{\partial x_k} +
  \rho \frac{\partial \theta_{ik}}{\partial x_k} +
  \theta_{ik} \frac{\partial \rho}{\partial x_k}
\right) \mathscr{H}^{[\theta]}_i (\bC) \\
& \qquad + \Bigg(
  \frac{1}{2} \rho u_k \frac{\partial \theta_{ij}}{\partial x_k} +
  \frac{1}{2} u_k \theta_{\langle ij \rangle}
    \frac{\partial \rho}{\partial u_k} +
  \rho \theta \frac{\partial u_i}{\partial x_j} \\
& \qquad \qquad +
  \frac{2}{5} \rho \theta_{\langle ik \rangle}
    \frac{\partial u_k}{\partial x_j} +
  \frac{1}{5} \rho \theta_{\langle kl \rangle} \delta_{ij}
    \frac{\partial u_k}{\partial x_l} +
  \frac{2}{5} \frac{\partial q_i}{\partial x_j} +
  \frac{1}{5} \delta_{ij} \frac{\partial q_k}{\partial x_k}
\Bigg) \mathscr{H}^{[\theta]}_{ij}(\bC) \\
& \qquad + \frac{1}{5} u_k \left(
  \rho \theta_{\langle jl \rangle} \frac{\partial u_l}{\partial x_k} +
  \frac{\partial q_j}{\partial x_k}
\right) \mathscr{H}^{[\theta]}_{iij}(\bC) + \frac{1}{2} \theta \left(
  \rho \frac{\partial \theta_{ij}}{\partial x_k} +
  \theta_{\langle ij \rangle} \frac{\partial \rho}{\partial x_k}
\right) \mathscr{H}^{[\theta]}_{ijk}(\bC) \\
& \qquad + \frac{1}{5} \theta \left(
  \rho \theta_{\langle jl \rangle} \frac{\partial u_l}{\partial x_k} +
  \frac{\partial q_j}{\partial x_k}
\right) \mathscr{H}^{[\theta]}_{iijk}(\bC).
\end{split}
\end{equation}
Projecting \eqref{eq:13m_convection} into
$\mathbb{S}_{13}^{[\bdeta]}$, one obtains
\begin{equation}
\begin{split}
\xi_k \frac{\partial f_{|13}}{\partial x_k} & \approx \left(
  u_k \frac{\partial \rho}{\partial x_k} +
  \rho \frac{\partial u_k}{\partial x_k}
\right) w^{[\theta]}(\bC) + \left(
  \rho u_k \frac{\partial u_i}{\partial x_k} +
  \rho \frac{\partial \theta_{ik}}{\partial x_k} +
  \theta_{ik} \frac{\partial \rho}{\partial x_k}
\right) \mathscr{H}^{[\theta]}_i (\bC) \\
& \qquad + \Bigg(
  \frac{1}{2} \rho u_k \frac{\partial \theta_{ij}}{\partial x_k} +
  \frac{1}{2} u_k \theta_{\langle ij \rangle}
    \frac{\partial \rho}{\partial u_k} +
  \rho \theta \frac{\partial u_i}{\partial x_j} \\
& \qquad \qquad +
  \frac{2}{5} \rho \theta_{\langle ik \rangle}
    \frac{\partial u_k}{\partial x_j} +
  \frac{1}{5} \rho \theta_{\langle kl \rangle} \delta_{ij}
    \frac{\partial u_k}{\partial x_l} +
  \frac{2}{5} \frac{\partial q_i}{\partial x_j} +
  \frac{1}{5} \delta_{ij} \frac{\partial q_k}{\partial x_k}
\Bigg) \mathscr{H}^{[\theta]}_{ij}(\bC) \\
& \qquad + \frac{1}{5} \left(
  u_k \rho \theta_{\langle jl \rangle}
    \frac{\partial u_l}{\partial x_k} +
  u_k \frac{\partial q_j}{\partial x_k} +
  \rho \theta \frac{\partial \theta_{jk}}{\partial x_k} +
  \frac{1}{2} \rho \theta \frac{\partial \theta_{kk}}{\partial x_j} +
  \theta\theta_{\langle jk \rangle} \frac{\partial \rho}{\partial x_k}
\right) \mathscr{H}^{[\theta]}_{iij}(\bC).
\end{split}
\end{equation}
Just like \eqref{eq:l_k}, the above approximation can be rewritten as
\begin{equation}
\xi_k \frac{\partial f_{|13}}{\partial x_k} \approx
  \boldsymbol{B}^T {\bf M}_k(\bw) {\bf D}(\bw)
  \frac{\partial \bw}{\partial x_k},
\end{equation}
where ${\bf M}_k(\bw)$, $k=1,2,3$ are matrices with size $13 \times
13$. Here we only give the expression of ${\bf M}_1$:
\begin{equation}
{\bf M}_1(\bw) = \left( \begin{array}{ccccccccccccc}
  u_1 & 1 & 0 & 0 & 0 & 0 & 0 & 0 & 0 & 0 & 0 & 0 & 0 \\
  \theta  & u_1 & 0 & 0 & 2 & 0 & 0 & 0 & 0 & 0 & 0 & 0 & 0 \\
  0 & 0 & u_1 & 0 & 0 & 0 & 0 & 1 & 0 & 0 & 0 & 0 & 0 \\
  0 & 0 & 0 & u_1 & 0 & 0 & 0 & 0 & 1 & 0 & 0 & 0 & 0 \\
  0 & \theta  & 0 & 0 & u_1 & 0 & 0 & 0 & 0 & 0 & 3 & 0 & 0 \\
  0 & 0 & 0 & 0 & 0 & u_1 & 0 & 0 & 0 & 0 & 1 & 0 & 0 \\
  0 & 0 & 0 & 0 & 0 & 0 & u_1 & 0 & 0 & 0 & 1 & 0 & 0 \\
  0 & 0 & \theta  & 0 & 0 & 0 & 0 & u_1 & 0 & 0 & 0 & 2 & 0 \\
  0 & 0 & 0 & \theta  & 0 & 0 & 0 & 0 & u_1 & 0 & 0 & 0 & 2 \\
  0 & 0 & 0 & 0 & 0 & 0 & 0 & 0 & 0 & u_1 & 0 & 0 & 0 \\
  0 & 0 & 0 & 0 & \frac{3 \theta}{5} & \frac{\theta}{5} &
    \frac{\theta}{5} & 0 & 0 & 0 & u_1 & 0 & 0 \\
  0 & 0 & 0 & 0 & 0 & 0 & 0 & \frac{\theta}{5} & 0 & 0 & 0 & u_1 & 0 \\
  0 & 0 & 0 & 0 & 0 & 0 & 0 & 0 & \frac{\theta}{5} & 0 & 0 & 0 & u_1
\end{array} \right),
\end{equation}
and its real diagonalizability can be directly verified\footnote{We
  verify the real diagonalizability by validating that $p({\bf M}_1) =
  0$ for $p(\lambda) = (\lambda-u_1)[5(\lambda-u_1)^2 - 7\theta] \cdot
  [5(\lambda-u_1)^4 - 26\theta (\lambda-u_1)^2 + 15\theta^2]$, which
  is the minimal polynomial of ${\bf M}_1$. Apparently $p(\lambda)$
  factors completely into distinct linear factor when $\theta >
  0$.}. Then the rotational invariance of the 13-moment expansion
yields the real diagonalizability of any linear combination of ${\bf
  M}_k(\bw)$, $k=1,2,3$. For Maxwell molecules, by expansion and
projection, we have the following approximation of the collision term:
\begin{equation}
\mathcal{S}(f) \approx \boldsymbol{B}^T \boldsymbol{q}(\bw)
  = -\frac{3\rho}{m_g} \chi^{(2,3)} \left(
    \frac{1}{2} \rho \theta_{\langle ij \rangle}
      \mathscr{H}^{[\theta]}_{ij}(\bC) +
    \frac{2}{15} q_j \mathscr{H}^{[\theta]}_{iij}(\bC)
  \right),
\end{equation}
where $\chi^{(2,3)}$ is a constant and we refer the readers to
\cite{Struchtrup} for its definition. Finally we have the following
globally hyperbolic 13-moment system:
\begin{equation}
\frac{\partial \bw}{\partial t} +
  [{\bf D}(\bw)]^{-1} {\bf M}_k(\bw) {\bf D}(\bw)
    \frac{\partial \bw}{\partial x_k} =
[{\bf D}(\bw)]^{-1} \boldsymbol{q}(\bw).
\end{equation}
In explicit formation, the system is as
\begin{equation}
\begin{split}
& \frac{\mathrm{d} \rho}{\mathrm{d} t} +
  \rho \frac{\partial u_k}{\partial x_k} = 0, \\
& \frac{\mathrm{d} u_i}{\mathrm{d} t} +
  \frac{\partial \theta_{ik}}{\partial x_k} +
  \frac{\theta_{ik}}{\rho} \frac{\partial \rho}{\partial x_k} = 0,
    \quad i = 1,2,3, \\
& \frac{\mathrm{d} \theta_{ij}}{\mathrm{d} t} -
  \theta_{ij} \frac{\partial u_k}{\partial x_k} +
  \frac{3}{5} \theta \left(
    \frac{\partial u_i}{\partial x_j} +
    \frac{\partial u_j}{\partial x_i} +
    \delta_{ij} \frac{\partial u_k}{\partial x_k}
  \right) + \frac{2}{5} \left(
    \theta_{ik} \frac{\partial u_k}{\partial x_j} +
    \theta_{jk} \frac{\partial u_k}{\partial x_i} +
    \delta_{ij} \theta_{kl} \frac{\partial u_k}{\partial x_l}
  \right) \\
& \qquad + \frac{2}{5\rho} \left(
    \frac{\partial q_i}{\partial x_j} +
    \frac{\partial q_j}{\partial x_i} +
    \delta_{ij} \frac{\partial q_k}{\partial x_k}
  \right) =
    -\frac{3\rho}{m_g} \chi^{(2,3)} \theta_{\langle ij \rangle},
    \quad i,j = 1,2,3, \\
& \frac{\mathrm{d} q_j}{\mathrm{d} t} -
  \theta_{\langle ij \rangle} \theta_{\langle ik \rangle}
    \frac{\partial \rho}{\partial x_k} -
  \rho \theta_{ij} \frac{\partial \theta_{ik}}{\partial x_k} +
  2\rho \theta \frac{\partial \theta_{ij}}{\partial x_i} +
  \frac{1}{2} \rho \theta \frac{\partial \theta_{kk}}{\partial x_j} =
    -\frac{2\rho}{m_g} \chi^{(2,3)} q_j, \quad j = 1,2,3.
\end{split}
\end{equation}

\subsubsection{Example 2: full moment theories}
In this section, we are going to reproduce the work of \cite{Fan_new}
and \cite{ANRxx}, where two types of globally hyperbolic moment
systems are derived. Here we adopt the following approximation to the
distribution function:
\begin{equation}
f(t,\bx,\bxi) \approx \tilde{f}(t,\bx,\bxi) =
  \sum_{\substack{\balpha \in \bbN^D\\ |\balpha| \leqslant M}}
    f_{\balpha}(t,\bx) \mathscr{H}_{\balpha}^{[{\bf\Theta}]}
    \left( \bC(t,\bx) \right),
\end{equation}
where
\begin{equation}
{\bf \Theta} = {\bf \Theta}^T = (\theta_{ij})_{D \times D}, \qquad
  \bC = \bxi - \bu(t,\bx).
\end{equation}
The basis functions are defined as
\begin{equation}
w^{[{\bf \Theta}]}(\bC) = \frac{1}{m_g \sqrt{\det(2\pi {\bf \Theta})}}
  \exp \left( -\frac{1}{2} \bC^T {\bf \Theta}^{-1} \bC \right), \qquad
\mathscr{H}_{\balpha}^{[{\bf \Theta}]}(\bC) =
  (-1)^{|\balpha|} \frac{\partial^{|\balpha|}}{\partial \bC^{\alpha}}
  w^{[{\bf \Theta}]}(\bC).
\end{equation}
Here $\bdeta = (u_1, \cdots, u_D, \theta_{11}, \cdots, \theta_{1D},
\theta_{22}, \cdots, \theta_{2D}, \cdots, \theta_{DD}),$
and thus $D(D+3)/2$ constraints are needed. These constraints will be
discussed later. We first define the inner product $\langle \cdot,
\cdot \rangle^{[\bdeta]}$ as
\begin{equation} \label{eq:ip}
\langle g_1, g_2 \rangle^{[\bdeta]} = \int_{\mathbb{R}^D}
  \frac{g_1(\bxi) g_2(\bxi)}{w^{[{\bf \Theta}]}(\bxi - \bu)}
\dd \bxi.
\end{equation}
Mimicing the abstract procedure, we calculate $\partial_{t,\bx} f$ and
perform the projection to space $\mathbb{S}_M^{[\bdeta]} =
\mathrm{span} \{ \mathscr{H}_{\balpha}^{[{\bf \Theta}]}(\bC) \mid
\balpha \in \bbN^D, |\balpha| \leqslant M \}$. The result is
\begin{equation} \label{eq:diff_expansion}
\frac{\partial \tilde{f}}{\partial s} \approx
  \sum_{\substack{\balpha \in \bbN^D\\ |\balpha| \leqslant M}} \left(
    \frac{\partial f_{\balpha}}{\partial s} +
    \sum_{i=1}^D \frac{\partial u_i}{\partial s} f_{\balpha-\be_i} +
    \frac{1}{2} \sum_{i,j=1}^D
      \frac{\partial \theta_{ij}}{\partial s} f_{\balpha-\be_i-\be_j}
  \right) \mathscr{H}^{[{\bf\Theta}]}_{\balpha}(\bC),
\end{equation}
where $s = t,x_1,\cdots,x_D$. Now we calculate the convection term and
apply the projection again:
\begin{equation}
\begin{split}
\xi_k \frac{\partial \tilde{f}}{\partial x_k} & \approx
\sum_{\substack{\balpha \in \bbN^D\\ |\balpha| \leqslant M}} 
  \Bigg[ u_k \left(
    \frac{\partial f_{\balpha}}{\partial x_k} +
    \sum_{i=1}^D \frac{\partial u_i}{\partial x_k} f_{\balpha-\be_i} +
    \frac{1}{2} \sum_{i,j=1}^D
      \frac{\partial\theta_{ij}}{\partial x_k} f_{\balpha-\be_i-\be_j}
  \right) \\
& \quad + (1 - \delta_{|\balpha|,M}) (\alpha_k + 1) \left(
    \frac{\partial f_{\balpha+\be_k}}{\partial x_k} +
    \sum_{i=1}^D \frac{\partial u_i}{\partial x_k}
      f_{\balpha-\be_i+\be_k} +
    \frac{1}{2} \sum_{i,j=1}^D \frac{\partial\theta_{ij}}{\partial x_k}
      f_{\balpha-\be_i-\be_j+\be_k}
  \right) \\
& \quad + \sum_{l=1}^D \theta_{kl} \left(
    \frac{\partial f_{\balpha-\be_l}}{\partial x_k} +
    \sum_{i=1}^D \frac{\partial u_i}{\partial x_k}
      f_{\balpha-\be_i-\be_l} +
    \frac{1}{2} \sum_{i,j=1}^D \frac{\partial\theta_{ij}}{\partial x_k}
      f_{\balpha-\be_i-\be_j-\be_l}
  \right) \Bigg] \mathscr{H}_{\balpha}^{[\bdeta]}(\bC), \\
& \hspace{11.5cm} k=1,\cdots,D.
\end{split}
\end{equation}
Supposing the projected right hand side has the following
approximation:
\begin{equation}
\mathcal{S}(f) \approx
  \sum_{\substack{\balpha \in \bbN^D\\ |\balpha| \leqslant M}}
    Q_{\balpha}(t,\bx) \mathscr{H}_{\balpha}^{[\bdeta]}(\bC),
\end{equation}
we then have the following moment system:
\begin{equation} \label{eq:full_moment}
\begin{split}
& \frac{\partial f_{\balpha}}{\partial t} +
\sum_{i=1}^D \frac{\partial u_i}{\partial t} f_{\balpha-\be_i} +
\frac{1}{2} \sum_{i,j=1}^D
  \frac{\partial\theta_{ij}}{\partial t} f_{\balpha-\be_i-\be_j} \\
& \quad + \sum_{k=1}^D \Bigg[ u_k \left(
    \frac{\partial f_{\balpha}}{\partial x_k} +
    \sum_{i=1}^D \frac{\partial u_i}{\partial x_k} f_{\balpha-\be_i} +
    \frac{1}{2} \sum_{i,j=1}^D
      \frac{\partial\theta_{ij}}{\partial x_k} f_{\balpha-\be_i-\be_j}
  \right) \\
& \quad \quad + (1 - \delta_{|\balpha|,M}) (\alpha_k + 1) \left(
    \frac{\partial f_{\balpha+\be_k}}{\partial x_k} +
    \sum_{i=1}^D \frac{\partial u_i}{\partial x_k}
      f_{\balpha-\be_i+\be_k} +
    \frac{1}{2} \sum_{i,j=1}^D \frac{\partial\theta_{ij}}{\partial x_k}
      f_{\balpha-\be_i-\be_j+\be_k}
  \right) \\
& \quad \quad + \sum_{l=1}^D \theta_{kl} \left(
    \frac{\partial f_{\balpha-\be_l}}{\partial x_k} +
    \sum_{i=1}^D \frac{\partial u_i}{\partial x_k}
      f_{\balpha-\be_i-\be_l} +
    \frac{1}{2} \sum_{i,j=1}^D \frac{\partial\theta_{ij}}{\partial x_k}
      f_{\balpha-\be_i-\be_j-\be_l}
  \right) \Bigg] = Q_{\balpha}, \\
& \hspace{11cm} \balpha \in \mathbb{N}^D, \quad |\balpha| \leqslant M.
\end{split}
\end{equation}

Before we check the hyperbolicity of this moment system, we have to
give the constraints \eqref{eq:constraints}. Below we consider two
cases, which correspond to the systems in \cite{Fan_new} and
\cite{ANRxx} respectively.
\subparagraph{First case: classic Hermite expansion} We require the
parameter $\bdeta$ and the coefficients to satisfy
\begin{equation}
\begin{split}
& f_{\be_j} = 0, \quad j = 1,\cdots,D,
  \qquad \qquad \sum_{j=1}^D f_{2\be_j} = 0, \\
& \theta_{11} = \theta_{22} = \cdots = \theta_{DD},
  \qquad \qquad \theta_{ij} = 0 \quad \text{if} \quad i \neq j.
\end{split}
\end{equation}
Define $\bw_m$ to be the vector whose components are all $f_{\balpha}$
with $|\balpha| = m$, and let
\begin{equation}
g_{\be_i + \be_j} = \frac{1}{2} \delta_{ij} \theta_{ij} +
  f_{\be_i + \be_j}, \quad i,j = 1,\cdots,D,
\qquad \qquad \bg = (g_{\balpha})|_{|\balpha| = 2}.
\end{equation}
Then the unknown vector $\bw$ can be chosen as
\begin{equation}
\bw = (f_{\bf 0}, u_1, \cdots, u_D, \bg, \bw_3, \cdots, \bw_M)^T.
\end{equation}
It is not difficult to observe from \eqref{eq:diff_expansion} that the
matrix ${\bf D}(\bw)$ (see eq. \eqref{eq:tilde_g}) is a triangular
matrix with all its diagonal entries nonzero, and thus the first
hyperbolicity condition (invertibility of $\bf D$) is fulfilled. The
second hyperbolicity condition (real diagonalizability of linear
combinations of ${\bf M}_k$) can actually be satisfied in a very
general case, as will be discussed in section \ref{sec:hyp_condition}.
We claim that in this case, the hyperbolic moment system is exactly
what is proposed in \cite{Fan_new}, which can be verified by direct
comparison of their explicit expressions.

\subparagraph{Second case: generalized Hermite expansion} We require
the parameter $\bdeta$ and the coefficients to satisfy
\begin{equation} \label{eq:gen_Hermite}
f_{\be_j} = 0, \quad j = 1,\cdots,D,
  \qquad \qquad f_{\balpha} = 0 \quad \text{if} \quad |\balpha| = 2,
\end{equation}
and define $\bw = (f_{\bf 0}, \bdeta, \bw_3, \cdots, \bw_M)^T$. In
this case, the matrix ${\bf D}(\bw)$ is also a triangular matrix with
all the diagonal entries nonzero, and thus is invertible. The real
diagonalizability of the linear combinations of ${\bf M}_k$ will also
be discussed in section \ref{sec:hyp_condition}, and we only claim
here that this condition is also fulfilled, which leads to the
global hypebolicity of \eqref{eq:full_moment}. With the constraints
\eqref{eq:gen_Hermite}, the above procedure reproduces the systems
derived in \cite{ANRxx}.

\subsection{Generic kinetic equation}
Here now on we are trying to extend the framework above to even
generic kinetic equation formulated as
\begin{equation} \label{eq:general_ke}
\pd{f}{t} + \bv(\bxi) \cdot \pd{f}{\bx} = \mathcal{S}(f).
\end{equation}
Such equation arises from such as Liouville equation of Hamiltonian
system, radiative transport or relativistic kinetic theory. For
example, let us consider a general Hamiltonian system with its
Hamiltonian to be $\mathcal{H}(t, \bx, \bxi)$. The probability density
function $f(t, \bx, \bxi)$ in the configuration space is governed by
the Liouville equation as
\begin{equation}\label{eq:liouville}
\pd{f}{t} + \pd{\cal H}{\bxi} \cdot \pd{f}{\bx} - \pd{\cal H}{\bx}
\cdot \pd{f}{\bxi} = 0,
\end{equation}
where the ordinate acceleration term $\partial_{\bx} {\cal H} \cdot
\partial_{\bxi} f$ is corresponding to $\mathcal{S}(f)$ in
\eqref{eq:general_ke}. Often no derivatives respected to $\bx$ is
involved in this term. The convective velocity $\bv(\bxi)$ in
\eqref{eq:general_ke} is then specified as $\partial_{\bxi}
\mathcal{H}$. For the generic kinetic equation \eqref{eq:general_ke},
due to its similarity to the Boltzmann equation, the procedure in
section \ref{sec:MD_Boltzmann} can be applied with only slightly
revision. Below we first give an outline of the deduction of the
hyperbolic system.
\begin{enumerate}
\item Expand the function $f$ into an infinite series
  \eqref{eq:general_expansion}, and specify the constraints
  \eqref{eq:constraints} on the coefficients and parameters.
\item For a sufficiently large $M$, use the constraints
  \eqref{eq:constraints} to eliminate $n$ items from $\eta_1, \cdots,
  \eta_n, F_0, F_1, \cdots, F_M$, and denote by $\bw$ the vector whose
  components are the remaining items.
\item Approximate $f$ by writing all $F_i$ with $i > M$ as functions
  of $\bw$ \eqref{eq:step1}.
\item Expand the time and spatial derivatives together with the right
  hand side into the series \eqref{eq:step2}.
\item Through the inner product $\langle \cdot, \cdot
  \rangle^{[\bdeta]}$, project \eqref{eq:step2} onto the finite
  dimensional space $\mathbb{S}_M^{[\bdeta]} = \{\varphi_0^{[\bdeta]},
  \cdots, \varphi_M^{[\bdeta]}\}$ (see \eqref{eq:step3}---\eqref
  {eq:tilde_g}), which provides us the approximation
  \[
  \pd{\tilde{f}}{s} \approx \left( \varphi_0^{[\bdeta(t,\bx)]}(\bxi),
    \cdots, \varphi_M^{[\bdeta(t,\bx)]}(\bxi) \right) {\bf D}(\bw)
  \frac{\partial \bw}{\partial s}, \qquad s = t,x_1, \cdots, x_D,
  \]
  and
  \[
  \mathcal{S}(\tilde{f}) \approx \sum_{i=0}^{M} \tilde{Q}_i(t,\bx)
  \varphi_i^{[\bdeta(t,\bx)]}(\bxi).
 \]
\item Calculate the convection term $v_k(\bxi) \partial_{x_k} f$ based
  on \eqref{eq:lhs}, and expand the result into series:
  \begin{equation} \label{eq:step4_new}
  v_k(\bxi) \frac{\partial \tilde{f}}{\partial x_k} \approx
    \sum_{i=0}^{+\infty} J_{k,i}(t,\bx)
    \varphi_i^{[\bdeta(t,\bx)]}(\bxi),
  \qquad k = 1,\cdots,D.
  \end{equation}
\item Project \eqref{eq:step4_new} onto $\mathbb{S}_M^{[\bdeta]}$:
  \begin{equation} \label{eq:step5_new}
  v_k(\bxi) \frac{\partial \tilde{f}}{\partial x_k} \approx
    \sum_{i=0}^M \tilde{J}_{k,i}(t,\bx)
      \varphi_i^{[\bdeta(t,\bx)]}(\bxi),
  \qquad k = 1,\cdots,D.
  \end{equation}
  Then by \eqref{eq:l_k} we also have
  \[
  \tilde{J}_{k,i} = {\bf M}_k(\bw) {\bf D}(\bw) \pd{\bw}{x_k}, \qquad
  k = 1,\cdots,D.
  \]
\item Put \eqref{eq:step5_new} and \eqref{eq:lhs} into
  \eqref{eq:general_ke}, and finally obtain the system
  as \eqref{eq:step6}, 
  \[
  {\bf D}(\bw) \frac{\partial \bw}{\partial t} + \sum_{k=1}^D {\bf
    M}_k(\bw) {\bf D}(\bw) \frac{\partial \bw}{\partial x_k} =
  \boldsymbol{q}(\bw),
  \]
  which is a reduced model for the generic kinetic equation
  \eqref{eq:general_ke}, where $\boldsymbol{q} = (\tilde{Q}_0, \cdots,
  \tilde{Q}_M)^T$.
\end{enumerate}
Again, the matrices ${\bf D}(\bw)$ and ${\bf M}_k(\bw)$ appear at the
same locations in the above deduction, and the system is hyperbolic
if ${\bf D}(\bw)$ is invertible and any linear combination of ${\bf
  M}_k(\bw)$ is real diagonalizable.

Below, we will first look into these two conditions leading to
hyperbolicity, and then provide the $M_N$ model for radiative
transport as the application of the framework.

\subsubsection{On the conditions to hyperbolicity}
\label{sec:hyp_condition}
In order to give the expressions of the matrices ${\bf D}(\bw)$ and
${\bf M}_k(\bw)$, we denote $\{\psi_0^{[\bdeta]}, \cdots,
\psi_M^{[\bdeta]}\}$ as an orthogonal basis of the $(M+1)$-dimensional
space $\mathbb{S}_M^{[\bdeta]}$, which satisfies
\begin{equation}
\langle \psi_i^{[\bdeta]}, \psi_j^{[\bdeta]} \rangle^{[\bdeta]}
  = \delta_{ij}, \qquad i,j=0,\cdots,M.
\end{equation}
Let $\boldsymbol{b}^{[\bdeta]} = (\varphi_0^{[\bdeta]}, \cdots,
\varphi_M^{[\bdeta]})^T$ and $\tilde{\boldsymbol{b}}^{[\bdeta]} =
(\psi_0^{[\bdeta]}, \cdots, \psi_M^{[\bdeta]})^T$. Clearly, the
orthogonal basis $\tilde{\boldsymbol{b}}^{[\bdeta]}$ can be obtained
by Schmidt's orthogonalization from $\boldsymbol{b}^{[\bdeta]}$. Then
there exists a non-singular matrix ${\bf T}^{[\bdeta]}$ such that
$\tilde{\boldsymbol{b}}^{[\bdeta]} = {\bf T}^{[\bdeta]}
\boldsymbol{b}^{[\bdeta]}$ \footnote{Precisely, the matrix ${\bf
    T}^{[\bdeta]}$ can be a lower triangular matrix provided by
  Schmidt's orthogonalization.}.  According to \eqref{eq:step1}, the
function $\tilde{f}$ can actually be written as a function of $\bw$
and $\bxi$, which is denoted by $\tilde{f}(t,\bx,\bxi) = F(\bw(t,\bx);
\bxi)$. Thus, the matrices ${\bf D}(\bw)$ and ${\bf M}_k(\bw)$ can be
written as
\begin{subequations}
\begin{align}
\label{eq:D}
& {\bf D} = {\bf T}^{[\bdeta]} \tilde{\bf D},
&& \tilde{\bf D} = \left \langle
  \tilde{\boldsymbol{b}}^{[\bdeta]},
  \left( \frac{\partial F}{\partial \bw} \right)^T
\right \rangle^{[\bdeta]}, \\
\label{eq:M_k}
& {\bf M}_k = {\bf T}^{[\bdeta]} \tilde{\bf M}_k
  \left( {\bf T}^{[\bdeta]} \right)^{-1}, \quad
&& \tilde{\bf M}_k = \left \langle
  \tilde{\boldsymbol{b}}^{[\bdeta]},
  v_k(\bxi) \left( \tilde{\boldsymbol{b}}^{[\bdeta]} \right)^T
\right \rangle^{[\bdeta]}, \qquad k=1,\cdots,D,
\end{align}
\end{subequations}
where the notation $\langle \tilde{\boldsymbol{a}}, \boldsymbol{a}^T
\rangle^{[\bdeta]}$ stands for the matrix with its entries as $\langle
\tilde{a}_i, a_j \rangle^{[\bdeta]}$. Through \eqref{eq:M_k}, we
immediately get the following criterion on the real diagonalizability
of ${\bf M}_k$:
\begin{theorem} \label{thm:real_diag}
Any linear combination of ${\bf M}_k$ is real diagonalizable if the
inner product $\langle \cdot, \cdot \rangle^{[\bdeta]}$ satisfies
\begin{equation} \label{eq:symmetric}
\langle v_k(\bxi) g_1, g_2 \rangle^{[\bdeta]} =
  \langle g_1, v_k(\bxi) g_2 \rangle^{[\bdeta]}, \qquad k = 1,\cdots,D
\end{equation}
for any $g_1, g_2 \in \mathbb{S}_M^{[\bdeta]}$.
\end{theorem}
\begin{proof}
Let $a_k$, $k=1,\cdots,D$ be $D$ arbitrary reals. By equation
\eqref{eq:M_k}, we know that $\sum_k a_k {\bf M}_k$ is real
diagonalizable if and only if $\sum_k a_k \tilde{\bf M}_k$ is real
diagonalizable. According to \eqref{eq:symmetric},
\begin{equation}
\begin{split}
\sum_{k=1}^D a_k \tilde{\bf M}_k^T &= \sum_{k=1}^D a_k \left \langle
  v_k(\bxi) \tilde{\boldsymbol{b}}^{[\bdeta]},
  \left( \tilde{\boldsymbol{b}}^{[\bdeta]} \right)^T
\right \rangle^{[\bdeta]} \\
& = \sum_{k=1}^D a_k \left \langle
  \tilde{\boldsymbol{b}}^{[\bdeta]},
  v_k(\bxi) \left( \tilde{\boldsymbol{b}}^{[\bdeta]} \right)^T
\right \rangle^{[\bdeta]} = \sum_{k=1}^D a_k \tilde{\bf M}_k,
\end{split}
\end{equation}
which means $\sum_k a_k \tilde{\bf M}_k$ is a symmetric matrix, and
thus is real diagonalizable.
\end{proof}
Obviously, both inner products \eqref{eq:13m_ip} and \eqref{eq:ip}
satisfy the hypothesis of the above theorem, and then the real
diagonalizability of linear combinations of ${\bf M}_k$ can be
directly obtained for both cases. The condition \eqref{eq:symmetric}
implies a quite natual symmetry on the inner product equipped by
$\mathbb{S}_M^{[\bdeta]}$. Clearly, an inner product formulated as
\[
\langle g_1, g_2 \rangle^{[\bdeta]} = \int g_1 g_2 w^{[\bdeta]}(\bxi)
\dd \bxi,
\]
for any weight function $w^{[\bdeta]}(\bxi)$, fulfils this condition,
while this formation actually includes almost every case of practical
interests. Thus the model derived by our framework is globally
hyperbolic for most situations, only if the the unknowns is properly
chosen that the matrix ${\bf D}$ is invertible.

Now we consider the matrix ${\bf D}$. Under this framework, the
unknown vector $\bw$ is adopted as the parameters to construct a
distribution function $F(\bw; \bxi)$ in $\mathbb{S}_M^{[\bdeta]}$ to
approximate the solution of the kinetic equation. Definitely, it is
not permitted that a distribution function $F(\bw; \bxi)$ is
represented by two different set of parameters $\bw$. Otherwise, the
uniqueness of the model is already destroyed at the very beginning ---
only due to the representation of the solution. Thus it is necessary
that
\begin{equation}\label{eq:cond_uniqueness}
\bw^0 \neq \bw^1 \Longrightarrow F(\bw^0; \bxi) \not\equiv F(\bw^1; \bxi).
\end{equation}
At a given $\bw^0$, we may linearize $F(\bw^1; \bxi)$ around $\bw^0$ to
have
\begin{equation} \label{eq:F_expansion}
F(\bw^1; \bxi) \approx F(\bw^0; \bxi) +
  \frac{\partial F}{\partial \bw}(\bw^0; \bxi) (\bw^1 - \bw^0).
\end{equation}
By \eqref{eq:D} and \eqref{eq:F_expansion}, one notice that
$\partial_{\bw} F$ is an approximation of the multiplication of ${\bf
D}$ and the basis functions, by dropping off the terms out of the
approximation space $\mathbb{S}_M^{[\bdeta]}$, which results in
\[
F(\bw^1; \bxi) \approx F(\bw^0; \bxi) +
\left(\boldsymbol{b}^{[\bdeta]}(\bxi) \right)^T {\bf D}(\bw^0) (\bw^1
- \bw^0).
\]
By \eqref{eq:cond_uniqueness}, it is natural to require the matrix
${\bf D}$ to be invertible. 

\begin{remark}
  However, sometimes the matrix ${\bf D}(\bw)$ may be singular on some
  isolated points while the mapping from $\bw$ to the projection of
  $F(\bw;\bxi)$ onto $\mathbb{S}_M^{[\bdeta]}$ remains injective. For
  example, we consider the case $D=1$ and only one basis function
  \begin{equation} \label{eq:example_F}
    F(\eta,\xi) = \varphi^{[\eta]}(\xi) = \exp(-\xi^2+\eta^3).
  \end{equation}
  Let $w = \eta$, and then $\partial_w F = 3\eta^2 F$. In this case,
  $\partial_w F = 0$ when $w = 0$, but the mapping $w \mapsto
  F(\eta,\xi)$ is an injection. Such singularity is often caused by
  the existence of saddle points in the mapping, and this can be
  removed by using alternative $\bw$ locally. In \eqref{eq:example_F}, if
  we choose $w=\eta^3$, then $\partial_w F$ becomes nonzero around
  $\eta = 0$.
\end{remark}

\newcommand\Mn{\texorpdfstring{$\boldsymbol{M_N}$}{M\_N}}
\subsubsection{Example 3: {\Mn} model for radiative transport}
Consider the radiative transfer equation
\begin{equation}
\frac{1}{c} \frac{\partial f}{\partial t} +
  \bv(\bxi) \cdot \frac{\partial f}{\partial \bx} =
  \mathcal{C}(f; T),
\end{equation}
where $c$ is the speed of light, the right hand side $\mathcal{C}$
models interactions between photons and the background medium with
$T$ as the material temperature, and
\begin{equation}
\begin{aligned}
& \bx = (x_1, x_2, x_3)^T, \qquad
\bxi = (\theta, \phi)^T \in [0,\pi] \times [0,2\pi), \\
& \bv(\bxi) = (v_1(\bxi), v_2(\bxi), v_3(\bxi))^T
  = (\sin\theta \cos\phi, \sin\theta \sin\phi, \cos\theta)^T.
\end{aligned}
\end{equation}
For any integer $N \geqslant 1$, define
\begin{equation}
\mathcal{M}^{[\bdeta]} = \left[ \exp\left(
  -\frac{\hbar \nu c}{k_{\rm B}} \sum_{m=0}^N \sum_{l=-m}^m \eta_{lm} Y_{lm}(\bxi)
\right) - 1 \right]^{-1},
\end{equation}
where $Y_{lm}(\bxi)$ are real spherical harmonics, and $\eta_{lm}$ are
the corresponding coefficients. The parameters $\hbar$, $\nu$ and
$k_{\rm B}$ denote Planck's constant, the frequency and Boltzmann's
constant, respectively. Then we define basis functions as
\begin{equation}
\varphi_{lm}^{[\bdeta]}(\bxi) = Y_{lm}(\bxi)
  \left( 1 + \mathcal{M}^{[\bdeta]} \right) \mathcal{M}^{[\bdeta]},
\qquad l=-m, \cdots, m,
\end{equation}
and the Hilbert spaces $\mathbb{H}^{[\bdeta]}$ and
$\mathbb{S}_N^{[\bdeta]}$:
\begin{subequations}
\begin{align}
& \mathbb{H}^{[\bdeta]} = \mathrm{span} \left\{
  \varphi_{lm}^{[\bdeta]} \,\Big|\, m=0,1,\cdots, \quad l=-m,\cdots,m
\right\}, \\
& \mathbb{S}_N^{[\bdeta]} = \mathrm{span} \left\{
  \varphi_{lm}^{[\bdeta]} \,\Big|\, m=0,1,\cdots,N, \quad l=-m,\cdots,m
\right\}, \\
\label{eq:me_ip}
& \langle g_1, g_2 \rangle^{[\bdeta]} = \int_0^{2\pi} \int_0^{\pi}
  g_1(\bxi) g_2(\bxi) \left[
    \left( 1 + \mathcal{M}^{[\bdeta]} \right) \mathcal{M}^{[\bdeta]}
  \right]^{-1} \dd \theta \dd \phi.
\end{align}
\end{subequations}
Using power series expansion, we have
\begin{equation}
\exp\left(
  \frac{\hbar \nu c}{k_{\rm B}} \sum_{m=0}^N \sum_{l=-m}^m \eta_{lm} Y_{lm}(\bxi)
\right) = \sum_{m=0}^{+\infty} \sum_{l=-m}^m
  a_{lm}^{[\bdeta]} Y_{lm}(\bxi).
\end{equation}
Now we approximate the function $f$ by
\begin{equation} \label{eq:max_entropy}
\begin{split}
\tilde{f}(t,\bx,\bxi) &=
  \frac{2 \hbar \nu^3}{c^2} \mathcal{M}^{[\bdeta(t,\bx)]}
= \frac{2 \hbar \nu^3}{c^2} \left[ 1 - \exp\left(
  \frac{\hbar \nu c}{k_{\rm B}} \sum_{m=0}^N \sum_{l=-m}^m \eta_{lm} Y_{lm}(\bxi)
\right) \right]
\left( 1+\mathcal{M}^{[\bdeta]} \right) \mathcal{M}^{[\bdeta]} \\
& = \frac{2 \hbar \nu^3}{c^2} (1-a_{00}^{[\bdeta(t,\bx)]})
    \varphi_{00}^{[\bdeta(t,\bx)]}(\bxi) -
  \sum_{m=0}^{+\infty} \sum_{l=-m}^m \frac{2 \hbar \nu^3}{c^2}
    a_{lm}^{[\bdeta(t,\bx)]} \varphi_{lm}^{[\bdeta(t,\bx)]}(\bxi)
  \in \mathbb{H}^{[\bdeta]}.
\end{split}
\end{equation}
We choose $\bw = \bdeta$ and the constraints are spontaneouly provided by
\eqref{eq:max_entropy}.

Taking time and spatial derivatives on both sides of
\eqref{eq:max_entropy}, we have
\begin{equation}
\begin{split}
\frac{\partial \tilde{f}}{\partial s} &=
  \frac{2 \hbar^2\nu^4}{k_{\rm B} c} \sum_{m=0}^N \sum_{l=-m}^m
  \frac{\partial \eta_{lm}}{\partial s} Y_{lm}(\bxi)
    \left( 1+\mathcal{M}^{[\bdeta]} \right) \mathcal{M}^{[\bdeta]} \\
& = \frac{2 \hbar^2\nu^4}{k_{\rm B} c} \sum_{m=0}^N \sum_{l=-m}^m
  \frac{\partial \eta_{lm}}{\partial s} \varphi_{lm}^{[\bdeta]}(\bxi),
\qquad s=t,x_1,x_2,x_3.
\end{split}
\end{equation}
which is already a function in $\mathbb{S}_N^{[\bdeta]}$ so that no
additional projection is needed. Meanwhile, it is clear that ${\bf D}
= \dfrac{2\hbar^2 \nu^4}{k_{\rm B} c} {\bf I}$ which is obviously invertible.

In order to obtain the hyperbolic system, one still needs to write the
matrices ${\bf M}_k(\bdeta)$ and project $\mathcal{C}(\tilde{f}; T)$
onto $\mathbb{S}_N^{[\bdeta]}$. All these calculations are routine and
here we omit the details which are quite tedious. We comment that the
resulting system is globally hyperbolic since the inner product
\eqref{eq:me_ip} satisfies the hypothesis of Theorem
\ref{thm:real_diag}, and the initial approximation
\eqref{eq:max_entropy} reveals that this moment system coincides with
the $M_N$ model \cite{Hauck}.


\section{Concluding remarks} \label{sec:conclusion} Through
investigation on a special type of one-dimensional hyperbolic models,
a general framework for the construction of hyperbolic models from the
kinetic equations has been established. We have shown that this
framework already covers a number of existing hyperbolic models (the
1D and $n$D hyperbolic regularizations of Grad's moment methods
\cite{Fan, Fan_new}, the globally hyperbolic moment systems with
generalized Hermite expansion \cite{ANRxx} and the $M_N$ model in
radiative transfer \cite{Hauck}), and can also be used to discover new
hyperbolic models. Actually, some more models such as the classic
discrete velocity model, the maximum entropy model \cite{Levermore}
and the quadrature-based projection method \cite{Koellermeier} are
also included. 

For a first order model with Cauchy data, the hyperbolicity is
necessary to the existence of the solution. Historically, the
hyperbolicity of different models is a quite subtle problem and has to
be analysed case by case with great patience, which eventually leads
to some prejudices against the moment method. It seems that the new
perception on the hyperbolicity of the moment equation in this paper
may bring us the dawn of hope to further development of the moment
method for kinetic equations.

\section*{Acknowledgements}
Cai was supported by the Alexander von Humboldt Research Fellowship
awarded by the Alexander von Humboldt Foundation of Germany. Li was
supported in part by the National Natural Science Foundation of China
No. 11325102 and No. 91330205.


\bibliographystyle{plain}
\bibliography{../article}
\end{document}